\documentclass[conference]{IEEEtran}
\IEEEoverridecommandlockouts
\usepackage{cite}
\usepackage{amsmath,amssymb,amsfonts}
\usepackage{algorithmic}
\usepackage{graphicx}
\usepackage{textcomp}
\usepackage{comment}
\usepackage{xcolor}
\usepackage{amsthm} 
\usepackage{amsfonts} 
\usepackage{subcaption}
\usepackage[numbers,sort]{natbib}
\def\BibTeX{{\rm B\kern-.05em{\sc i\kern-.025em b}\kern-.08em
    T\kern-.1667em\lower.7ex\hbox{E}\kern-.125emX}}
\usepackage[labelformat=empty]{caption}
\begin{document}

\newtheorem{Definition}{Definition}
\newtheorem{Theorem}{Theorem}
\newtheorem{Lemma}{Lemma}

\title{
Differential Privacy Preserving Quantum Computing via Projection Operator Measurements
}

\author{
\IEEEauthorblockN{Yuqing Li$^*$, Yusheng Zhao$^*$, Xinyue Zhang$^\dagger$, Hui Zhong$^\ddagger$, Miao Pan$^\ddagger$, Chi Zhang$^*$}
\IEEEauthorblockA{$^*$\textit{School of Cyberspace Science and Technology},
\textit{University of Science and Technology of China}, Hefei,  P. R. China}
\IEEEauthorblockA{$^\dagger$\textit{Department of Computer Science},
\textit{Kennesaw State University}, Marietta, USA}
\IEEEauthorblockA{$^\ddagger$\textit{Department of Electrical and Computer Engineering},
\textit{University of Houston}, Houston, USA}
\IEEEauthorblockA{\{lyq62693641, yushengzhao\}@mail.ustc.edu.cn, xzhang48@kennesaw.edu, hzhong5@uh.edu, \\mpan2@central.uh.edu, chizhang@ustc.edu.cn}
}

\maketitle

\begin{abstract}
Quantum computing has been widely applied in various fields, such as quantum physics simulations, quantum machine learning, and big data analysis. However, in the domains of data-driven paradigm, how to ensure the privacy of the database is becoming a vital problem. For classical computing, we can incorporate the concept of differential privacy (DP) to meet the standard of privacy preservation by manually adding the noise. In the quantum computing scenario, researchers have extended classic DP to quantum differential privacy (QDP) by considering the quantum noise. In this paper, we propose a novel approach to satisfy the QDP definition by considering the errors generated by the projection operator measurement, which is denoted as shot noises. Then, we discuss the amount of privacy budget that can be achieved with shot noises, which serves as a metric for the level of privacy protection. Furthermore, we provide the QDP of shot noise in quantum circuits with depolarizing noise. Through numerical simulations, we show that shot noise can effectively provide privacy protection in quantum computing.

\end{abstract}


\section{Introduction}
Quantum computing brings new opportunities to the high performance computing family and has already shown excellent potential in several fields, such as drug research \cite{cao2018potential}, materials engineering \cite{bauer2020quantum}, and financial investment \cite{orus2019quantum}. Compared with traditional computing, quantum computing handles problems more efficiently, especially when dealing with large data sets and complex operations. With the development of quantum computing, a variety of quantum algorithms have emerged, including Grover's algorithm \cite{grover1996fast} and Shor's algorithm \cite{shor1994algorithms}, which demonstrate significant quantum advantages. However, due to the sensitive datasets that will be used in quantum computing, the public is always worried about information leakage in the quantum computing process. Therefore, how to ensure the privacy of the data in quantum computing becomes a crucial issue.

\par Differential privacy (DP), which was first proposed by Dwork et al. in~\cite{dwork2014algorithmic}, has evolved into a reliable privacy preserving method. As its strict mathematical definition, DP guarantees that the results of an algorithm do not reveal any individual's information with minor alterations. In the classical scenario, one of the most common DP mechanisms is to add Gaussian noise to the statistical information extracted from a private dataset. Inspired by the wide applications of classical DP, researchers aim to propose the concept of quantum differential privacy (QDP) to protect quantum algorithms \cite{zhou2017differential, hirche2023quantum, watkins2023quantum, Du_2021Quantum} against privacy leakage. However, quantum algorithms are all based on the quantum states, which means that we cannot simply introduce noise (e.g., Gaussian noise \cite{dwork2014algorithmic}) as we did in the classical computing.

\par As a result, we need to explore alternative approaches to achieve DP in quantum computing. Researchers have recently extended this concept by leveraging quantum circuit noises to achieve QDP. 
QDP is initially proposed by \cite{zhou2017differential} and has subsequently been experimented on simulated quantum devices \cite{watkins2023quantum}. Due to the presence of quantum noise in current quantum devices, particularly in Noisy Intermediate-Scale Quantum (NISQ) devices, researchers investigate how quantum algorithms executed on these noisy circuits can achieve DP guarantee. In previous research, the focus was primarily on studying the privacy budget resulting from quantum noise, such as depolarizing noise and amplitude/phase damping noise~\cite{zhou2017differential,guan2023detecting}. 

However, these previous works have ignored an important aspect, which is the impact of measurement errors. In this paper, we define it as shot noise. We believe that implementing QDP through the amplification of shot noise is easier than introducing additional quantum noise in the circuit. Therefore, in this paper, we will address and discuss the privacy budget consumed during the measurement, which serves as a metric for the level of privacy protection. Measurement is described by measurement operators. In various measurement operators applied in quantum computing, projection operators are typical and practical, which are often used in diverse quantum tasks \cite{ivanova2020evaluating}. Regarding the discussion of shot noise, Wecker and Wandzura in~\cite{Wecker2015Dave} utilized the sampling theorem to model the measurement error in their work.
Different from their work, we further give a QDP privacy budget based on the shot noise measured from the projected observations.
Our prominent contributions are summarized as follows:
\begin{itemize}
    \item We find that shot noise can achieve QDP, which is also an important noise source of quantum computing. In a sense, it is an easier way to achieve the QDP by just reducing the number of shots since the shot noise increases with less number of shots. Moreover, we discuss both the privacy budget in the noisy and noiseless quantum circuits.
    \item We conduct numerical simulations to show the impacts of key parameters on the proposed QDP privacy budget. These parameters include the distance between two quantum databases, the number of shots, and the level of quantum noise in the circuit.
\end{itemize}

The rest of the section is as follows. In section \uppercase\expandafter{\romannumeral 2}, we present the background of quantum computing, the formal definition of QDP and emphasize the significance of shot noise in it. In section \uppercase\expandafter{\romannumeral 3}, we formalize the model, analyze the application of QDP in quantum circuits with depolarizing noise, and discuss the effects of shot noise brought by projection operators. In section \uppercase\expandafter{\romannumeral 4}, we provide numerical simulations of our results. Finally, we draw conclusions and discuss future work.
\section{PRELIMINARIES}

In this section, we introduce quantum computing background, quantum measurement and the basic formulas of classical DP and QDP.  

\subsection{Quantum Computing}

The notion of qubit is proposed to represent quantum states, and its effect is analogous to classical bits in classical computers. Single-qubit is commonly written using Dirac notation as:
\begin{equation}
	|\psi\rangle = 
    \begin{pmatrix}
		a\\b
    \end{pmatrix} 
    = a |0\rangle + b |1\rangle,
\mbox{where} \ |a|^2 + |b|^2 = 1.
\end{equation}

In quantum computations, quantum systems can be classified into pure states and mixed states, based on whether they can be represented in a tensor product form. For example, in multi-qubit systems, especially those interacting with the environment, it is common for the subsystems to be in a mixed state. To conveniently represent the mixed state, the concept of density operator is introduced, defined as follows:
\begin{equation}
    {\rho}= \sum_k p_k | {\psi}_k \rangle \langle
    {\psi}_k |,
\end{equation}
where the $\{\psi_k \}$ are the ensemble of $\rho$.

Quantum computation describes the manipulation and transformation of quantum states. Like the classical computers that are built on electrical circuits consisting of logic gates, a quantum algorithm is based on a quantum circuit
containing various quantum gates \cite{nielsen2010quantum}.

We have depicted the basic circuit structure of quantum computing in Fig. \ref{fig:quantum_circuit}. The orange section is a quantum encoder, which transfers classical data to quantum data. The blue section is a quantum circuit with quantum gate. The circuits of quantum computing mainly include the two parts. On the rightmost side of Fig. \ref{fig:quantum_circuit}, the measurement provides information about the final quantum state.
\begin{figure}[htbp]
	\centering
	\includegraphics[width=0.475\textwidth]{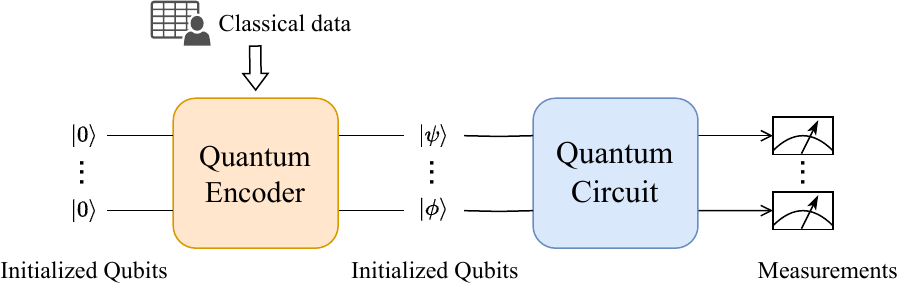}
	\caption{ Fig. 1: The basic structure of quantum computing.}
	\label{fig:quantum_circuit}
\end{figure}
\par A quantum circuit involves a series of quantum logic gates that perform unitary transformations on the quantum states, i.e., $U = U_1 U_2 \ldots U_n$. Each gate corresponds to a specific unitary matrix $U_i$,  which allows us to express the transformation as $| {\psi}' \rangle = U | {\psi} \rangle$ for a pure state and ${\rho}' = U^\dagger {\rho}U$ for a mixed state. 

\subsection{Quantum Measurement}
To provide a clearer explanation of shot noise, we will elaborate on the concept of quantum measurement. Quantum measurement is indeed a crucial step in the whole quantum computation \cite{Schlosshauer2013Snapshot}, which allows us to extract information from quantum systems. When a measurement is performed on a quantum system, the quantum state collapses into one of the eigenstates of the observable. The outcomes of quantum measurements are probabilistic,  so to improve the accuracy, we used to perform multiple measurements on the circuit \cite{Wecker2015Dave}. For quantum computing, it is common practice to obtain the expectation value of these measuring outputs and use it as the computation outcome \cite{Torlai2020Giacomo}.

Measurements are described by observable and the measurement operator is Hermitian, which satisfies $M^\dagger = M$. One of the important measurement methods is the \textit{Projection-Valued Measure (PVM)}, and we describe it by an observable $M$. The measurement expectation of the observable $M$ can be represented by the following equation:
\begin{equation}
    \begin{aligned}
	    \langle M \rangle &= \sum_{j} p_j \langle \psi_j | M | \psi_j \rangle \\
     &= \sum_{j} p_j \text{Tr}(|\psi_j \rangle \langle \psi_j| M) \\
        &= \text{Tr}\left(\sum_{j} p_j |\psi_j \rangle \langle \psi_j| M \right)= \text{Tr}(\rho M).
    \end{aligned}
\end{equation}

The projection operator is the simplest measurement operator in projection measurement and is widely used in quantum computing. We denote the projection operator as $\{M_i\}$, where the eigenvalues of $M_i$ are 0 and 1. The eigenvalue 1 indicates that a vector is projected within the subspace, while the eigenvalue 0 indicates that a vector is projected outside the subspace. For example, the rank of 2-dimensional projection operators $| 0 \rangle \langle 0 |$ and $| 1 \rangle \langle 1 |$ are both 1, meaning that it projects a quantum bit onto a one-dimensional subspace, i.e.,
$\begin{pmatrix}a\\b\end{pmatrix} \rightarrow  \begin{pmatrix}a\\0\end{pmatrix}$ or 
$\begin{pmatrix}a\\b\end{pmatrix} \rightarrow  \begin{pmatrix}0\\b\end{pmatrix}$.

\begin{figure*}[ht]
	\centering
     \includegraphics[width=1\textwidth]{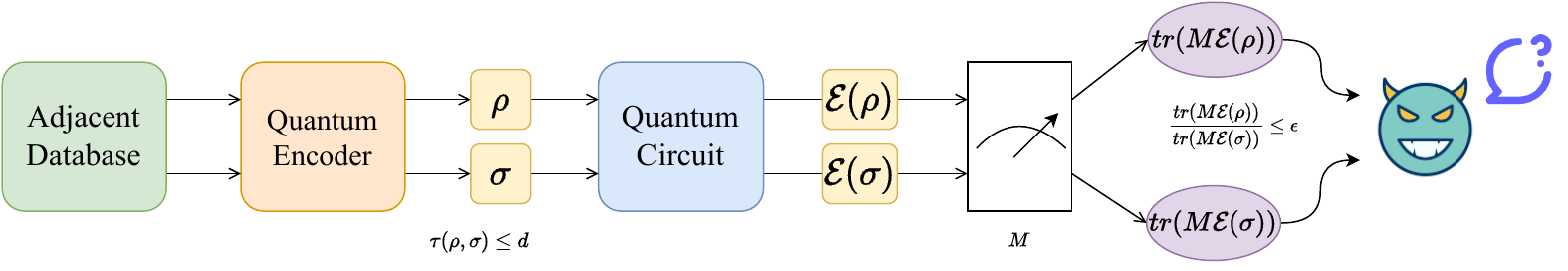}
	\caption{Fig. 2: Overview of the proposed quantum differential privacy.}
	\label{fig:dp_model}
\end{figure*}

\subsection{From classical to quantum differential privacy}
DP has evolved as a promising method for assessing an algorithm’s capacity for privacy preservation. 
We can define classical DP as follows.
\begin{Definition}[Classical Differential Privacy \cite{dwork2014algorithmic}]
	Supposed $\mathcal{M}$ is a randomized mechanism that takes as input entries $E\in \mathcal{D}$, and $\mathcal{S}$ is a set of possible outputs for the mapping $\mathcal{M}(\mathcal{D})$.
    Then  $\mathcal{M}$ gives $(\epsilon,\delta)$-DP if for a distance metric $h(\cdot,\cdot)$ and all data sets $E,E'\in \mathcal{D}$ satisfying $h(E,E')\leq 1$, all possible outcomes $\mathcal{S} \subseteq Range(\mathcal{M})$, we have 
	\begin{equation}
		\mathrm{Pr}[\mathcal{M}(E)\in \mathcal{S}] \leq \exp(\epsilon) \cdot \mathrm{Pr}[\mathcal{M}(E')\in \mathcal{S}]  + \delta.
	\end{equation}
    $\mathcal{M}$ is said to satisfy $\epsilon$-differential privacy if $\delta=0$.
    \label{def:cdp}
\end{Definition}

To generalize the DP definition into the quantum scenario, trace distance is introduced as a common distance metric that measures the closeness between two quantum states \cite{nielsen2010quantum}.
For two quantum states $\rho$ and $\sigma$ which represented by density matrices, we have
\begin{equation}\label{distance}
    \tau(\rho, \sigma) = \frac{1}{2}\text{Tr}(|\rho-\sigma|).
\end{equation}  
Several papers \cite{zhou2017differential, hirche2023quantum} have demonstrated the definition of QDP as follows.

\begin{Definition}[Quantum Differential Privacy \cite{zhou2017differential}]
    Supposed $\mathcal{A}=(\mathcal{E},\{M_k\}_{k\in\mathcal{S}})$ is a 
    quantum mechanism on a Hilbert space $\mathcal{H}$, where $\mathcal{E}$ is a quantum circuit and $\{M_k\}_{k\in\mathcal{S}}$ is the projection operator for measurement, 
    we are given two quantum state $\rho$ and $\sigma$ encoded from adjacent database 
    and a distance $\tau(\rho,\sigma)\leq d$. The mechanism $\mathcal{A}$ satisfies $(\epsilon,\delta)$-quantum differential privacy if we have
    \begin{equation}
        \sum_{k\in\mathcal{S}}\mathrm{Tr}[M_k\mathcal{E}(\rho)] \leq \exp{(\epsilon)} \cdot \sum_{k\in\mathcal{S}}\mathrm{Tr}[M_k\mathcal{E}(\sigma)] + \delta.
    \end{equation}
    $\mathcal{A}$ is said to satisfy $\epsilon$-quantum differential privacy if $\delta=0$.
\end{Definition}
\section{Quantum differential privacy based on shot noise}

In this section, we will derive new formulas for privacy budget $\epsilon$ in QDP by utilizing shot noise. The complete mathematical derivations can be found in the appendix. We present the overview of the proposed QDP model in Fig.~\ref{fig:dp_model}. We discuss the ${\epsilon}$-differential privacy and extend it to $(\epsilon,\delta)$-differential privacy under two quantum circuits separately. Firstly, we assume a noiseless circuit, where all the privacy budget arises from shot noise. Secondly, we assume a noisy circuit and use depolarizing noise to represent the incoherent noise in quantum circuits.

\begin{Theorem}
    Let $n$ be the number of shots of the quantum circuit. $\{M_i\}$ is a set of projection operators. Assuming that the rank of $M_m$ is the largest and considering it as r. For all quantum data $\rho$ and $\sigma$ such that ${\tau} ({\rho}, {\sigma}) \leqslant d$  and assume $\sigma'=\$(\sigma) $, where $\$ $ is the super-operator and then the $\sigma'$ represent the quantum state to be measured. Therefore, the outcomes can be written as $\mu_0={Tr}({\rho'M_m})$ and $\mu_1={Tr}({\sigma'M_m})$ of the adjacent quantum states. It holds that the measurement of $\{M_i\}$ satisfies ${\epsilon}$-differential privacy with 
    \begin{equation}\label{eq:theorem1}
    {\epsilon}= \frac{{dr}}{(1 -{\mu}) {\mu}}
    \left[ \frac{9}{2}(1 - 2{\mu})  + \frac{3}{2}\sqrt{n}+ \frac{{dr} ({\mu} + {dr})n}{1 -{\mu}}\right],
    \end{equation}
    where $\mu=min\{\mu_1,\mu_0\}$. 
\end{Theorem} 

\begin{proof}
We denote ${\rho}' ={\mathcal{E}} ({\rho})$ and ${\sigma}' ={\mathcal{E}} ({\sigma})$ as shown in Fig. \ref{fig:dp_model}. From \cite{zhou2017differential}, we know that if ${\tau} ({\rho}, {\sigma}) \leqslant d$, then ${\tau} ({\rho}', {\sigma}') \leqslant d$ and ${Tr} ({\rho}' -{\sigma}') M_m \leqslant {dTr}
(M_m) = {dr}$.

Since we have ${\mu}_0 = {Tr} (M_m {\rho}')$ and ${\mu}_1 = {Tr} (M_m {\sigma}')$, then we can derive ${\mu}_0 -{\mu}_1 = {Tr} [({\rho}'-{\sigma}') M_m] \leqslant {dr}$. According to the properties of projection measurements, we know that the ${\mu}_0$ represents the probability of obtaining a measurement result of 1, and $1-{\mu}_1$ represents the probability of obtaining a measurement result of 0. Therefore, it is easy to calculate the variance of the single measurement for ${\rho}'$ as ${\sigma}^2_0 = {\mu}_0 (1 -{\mu}_0)$. Similarly, the variance of the single  measurement for ${\sigma}'$ is  ${\mu}_1 (1 -{\mu}_1)$.

We can apply the central limit theorem~\cite{billingsley2017probability} to model the error of multiple measurements of the circuit. This allows us to obtain the distribution of the measuring average value, which follows a normal distribution. This distribution accurately represents the probability distribution of the circuit's outcomes, i.e. ${\mu}_0 $ and ${\mu}_1 $. Hence, the results obtained from measuring the circuit with $M_m$ respectively follow individual normal distributions $\mathcal{N} \left( {\mu}_0, \frac{{\mu}_0 (1
	-{\mu}_0)}{n} \right)$, $\mathcal{N} \left( {\mu}_1, \frac{{\mu}_1 (1
	-{\mu}_1)}{n} \right)$. The following derivation is similar to that of Gaussian noise, with the main distinction being that the variance of Gaussian noise is the same, whereas in this case, it is different. Assuming that $\mu_1=min\{\mu_0,\mu_1\}$, we can obtain the following expression,
 \begin{equation}
  \begin{split}
        \ln \frac{p  ({\rho}')}{p ({\sigma}')}  =&n ({\mu}_0 -{\mu}_1) \left[ {\frac{1 -{\mu}_0
                -{\mu}_1}{2{\mu}_0 {\mu}_1 (1 -{\mu}_0) (1
                -{\mu}_1)} x^2 +}\right. \\ &\left.{ \frac{1}{(1 -{\mu}_0) (1
                -{\mu}_1)} x - \frac{1}{2 (1 -{\mu}_0) (1
                -{\mu}_1)} }\right]\\
            \leqslant  &{ndr} \left[{ \frac{1 -{\mu}_0
                -{\mu}_1}{2{\mu}_0 {\mu}_1 (1 -{\mu}_0) (1
                -{\mu}_1)} x^2 +} \right. \\ &\left.{
                \frac{1}{(1 -{\mu}_0) (1
                -{\mu}_1)} x - \frac{1}{2 (1 -{\mu}_0) (1
                -{\mu}_1)}} \right] \\
                \leqslant &\epsilon.
   \end{split} \label{equ:8}
 \end{equation}

  Given that $x$ represents the possible outcome of the circuit, and a single measurement can only result in either $1$ or $0$, the range of values for $x$ is restricted to $(0, 1)$. To calculate $\epsilon$-differential privacy, we use the 3-sigma rule to approximate. This approximation guarantees that both ${\mu}_0 - 3{\sigma}_0$ and ${\mu}_0 + 3{\sigma}_0$ satisfy the Eq. (\ref{equ:8}), that is,
        \begin{align}
        &{ndr} \left[{ \frac{1 -{\mu}_0
            -{\mu}_1}{2{\mu}_0 {\mu}_1 (1 -{\mu}_0) (1
            -{\mu}_1)} ({\mu}_0 - 3{\sigma}_0)^2 + }\right. \nonumber\\ &\left.{\frac{1}{(1
            -{\mu}_0) (1 -{\mu}_1)} ({\mu}_0 -
        3{\sigma}_0) - \frac{1}{2 (1 -{\mu}_0) (1
            -{\mu}_1)}} \right] \leqslant {\epsilon},\\
        &{ndr} \left[ {\frac{1 -{\mu}_0
            -{\mu}_1}{2{\mu}_0 {\mu}_1 (1 -{\mu}_0) (1
            -{\mu}_1)} ({\mu}_0 + 3{\sigma}_0)^2 + }\right. \nonumber\\ &\left.{ \frac{1}{(1
            -{\mu}_0) (1 -{\mu}_1)} ({\mu}_0 +
        3{\sigma}_0) - \frac{1}{2 (1 -{\mu}_0) (1
            -{\mu}_1)} }\right] \leqslant {\epsilon}.
        \end{align}
Consequently, we can further scale the inequality to derive the final expression as Eq.~\eqref{eq:theorem1}.
\end{proof}
\addtolength{\topmargin}{0.02in}
In general tasks, a large number of shots are often used to attain high accuracy of measuring. However, to satisfy the standard of privacy protection, we can set shots small. This is not contradictory to the usual practice of taking large shots.

Due to the unavoidable noise in current quantum devices, we will also discuss ${\epsilon}$-differential privacy in the noisy circuit. Depolarizing noise is an important type of quantum noise which can be written as:
\begin{equation}
\mathcal{E}_{{Dep}} ({\rho}) = (1 - p) {\rho}+
\frac{p}{D} I,
\end{equation}
where $\rho$ is the density operator, $p$ represents the probability of depolarizing, $D$ is the dimension of the Hilbert space and $I$ is D dimensional identity matrix. 

Briefly speaking, with depolarizing noise in the circuit, we can get a tighter upper bound of ${\mu}_0-{\mu}_1$, which can be mentioned as,
\begin{equation}
        \begin{aligned}
	&\frac{{\mu}_0}{{\mu}_1} \leqslant 1 + \frac{1 - p}{p}
	{dD}\\
	\Leftrightarrow &{\mu}_0 -{\mu}_1 \leqslant \frac{1 - p}{p}
	{\mu}_1 {dD}.
\end{aligned}\label{equ:12}
    \end{equation}  
 Therefore, we put Eq. (\ref{equ:12}) to Eq. (\ref{equ:8}) and get the privacy budget in the noisy circuit.

\begin{Theorem}
    Let $n$ be the number of shots of the quantum circuit. $\{M_i\}$ is a set of projection operators. Assuming that the rank of $M_m$ is the largest and considering it as r. For all quantum data $\rho$ and $\sigma$ such that ${\tau} ({\rho}, {\sigma}) \leqslant d$  and assume $\sigma'=\$(\sigma) $, where $\$ $ is the super-operator and then the $\sigma'$ represent the quantum state to be measured. Therefore, the outcomes can be written as $\mu_0={Tr}({\rho'M_m})$ and $\mu_1={Tr}({\sigma'M_m})$ of the adjacent quantum states. If we add depolarizing noise to the circuit, then the privacy budget of $\epsilon$-differential privacy is,
	\begin{equation}
{\epsilon}= \frac{\alpha}{1 -{\mu}}
\left[ \frac{9}{2} (1 - 2{\mu}) + \frac{3}{2} \sqrt{n} 
\\
+\frac{\alpha {\mu}^2 \left( 1 + \alpha \right)}{1 -{\mu}} n \right],
	\end{equation}
where $\mu=min\{\mu_0,\mu_1\}$, $p$ denotes the probability of depolarizing, $D$ denotes the dimension of Hilbert space and $\alpha=\frac{1-p}{p}drD$.
\end{Theorem}

In addition to ${\epsilon}$-differential privacy, $(\epsilon,\delta)$-differential privacy is also an important concept in classical DP. It allows the situation to exist, where the algorithm may not satisfy DP with a certain probability.

Similar to our previous discussion, we will initially discuss it in the noiseless circuit and then delve into the condition in the noisy circuit.

To obtain the relation between $\delta$ and $\epsilon$,  We apply the notion of the Error function \cite{Andrews1991Special},
\begin{equation}
        \begin{aligned}
	{erf} (x) = \frac{1}{\sqrt{{\pi}}} \int^x_{- x} e^{- t^2}
	{dt} = \frac{2}{\sqrt{{\pi}}} \int^x_0 e^{- t^2} {dt},\\
	{erfc} (x) = 1 - {erf} (x) = \frac{2}{\sqrt{{\pi}}}
	\int^{\infty}_x e^{- t^2} {dt}.
\end{aligned}
    \end{equation}
We have
\begin{equation}
\begin{aligned}
{\delta} &=  \Pr \{ | x -{\mu} | > c \} \\
&= \int^{\infty}_{c
	+{\mu}} e^{- \frac{(m -{\mu})^2}{2{\sigma}^2}}
{dm} + \int^{{\mu}- c}_{- \infty} e^{- \frac{(m
		-{\mu})^2}{2{\sigma}^2}} {dm}.
\end{aligned}
\end{equation}
Assume that
\begin{equation*}
t = \frac{m -{\mu}}{\sqrt{2} {\sigma}},
\end{equation*}
then,
\begin{equation}
        \begin{aligned}
	{\delta}&= \Pr \{ | x -{\mu} | > c \}\\
	&=\sqrt{2} {\sigma} \int^{\infty}_{\frac{c}{\sqrt{2}
			{\sigma}}} e^{- t^2} {dt} + \sqrt{2} {\sigma} \int^{-
		\frac{c}{\sqrt{2} {\sigma}}}_{- \infty} e^{- t^2} {dt}\\
	&=2 \sqrt{2} {\sigma} \int^{\infty}_{\frac{c}{\sqrt{2}
			{\sigma}}} e^{- t^2} {dt}\\
	&=\sqrt{2{\pi}} {\sigma} {erfc} \left( \frac{c}{\sqrt{2}
			{\sigma}} \right).
\end{aligned}
    \end{equation}

\begin{Theorem}
Let $n$ be the number of shots of the quantum circuit. $\{M_i\}$ is a set of projection operators. Assuming that the rank of $M_m$ is the largest and considering it as r. For all quantum data $\rho$ and $\sigma$ such that ${\tau} ({\rho}, {\sigma}) \leqslant d$  and assume $\sigma'=\$(\sigma) $, where $\$ $ is the super-operator and then the $\sigma'$ represent the quantum state to be measured. Therefore, the outcomes can be written as $\mu_0={Tr}({\rho'M_m})$ and $\mu_1={Tr}({\sigma'M_m})$ of the adjacent quantum states. It is holds that the measurement of $\{M_i\}$ satisfy $(\epsilon,\delta)$-differential privacy, 
	\begin{equation}
{\epsilon}= \frac{{ndr}}{{\mu} (1 -{\mu})} \left[ \frac{(1 -
2{\mu} - {ndr}) c^2}{2{\mu} (1 -{\mu} -
{ndr})} + c + \frac{{ndr}}{2} \right],
	\end{equation}
where $\mu=min\{\mu_0,\mu_1\}$, ${\delta}= \sqrt{2{\pi}} {\sigma} {erfc}
\left( \frac{c}{\sqrt{2} {\sigma}} \right)$, ${\sigma}= \sqrt{\frac{{\mu} (1 -{\mu})}{n}}$.
\end{Theorem}

In the noisy circuit, the derivation of $\epsilon$ and $\delta$ is similar to the noiseless circuit.
\begin{Theorem}
Let $n$ be the number of shots of the quantum circuit. $\{M_i\}$ is a set of projection operators. Assuming that the rank of $M_m$ is the largest and considering it as r. For all quantum data $\rho$ and $\sigma$ such that ${\tau} ({\rho}, {\sigma}) \leqslant d$  and assume $\sigma'=\$(\sigma) $, where $\$ $ is the super-operator and then the $\sigma'$ represent the quantum state to be measured. Therefore, the outcomes can be written as $\mu_0={Tr}({\rho'M_m})$ and $\mu_1={Tr}({\sigma'M_m})$ of the adjacent quantum states. If we add depolarizing noise to the circuit, then it will provide $(\epsilon,\delta)$-differential privacy, 
	\begin{equation}
{\epsilon}= \frac{{\alpha}}{1 -{\mu}} \left[
\frac{(1 - 2{\mu} - n{\alpha}) c^2}{2{\mu} (1
-{\mu} - n{\alpha})} + c + \frac{n{\alpha}}{2}
\right],
	\end{equation}
where $\mu=min\{\mu_0,\mu_1\}$, ${\alpha}= \frac{1 - p}{p}drD$, ${\delta}= \sqrt{2{\pi}} {\sigma} {erfc}
\left( \frac{c}{\sqrt{2} {\sigma}} \right)$, ${\sigma}= \sqrt{\frac{{\mu} (1 -{\mu})}{n}}$, and $p$ is noise parameter.
\end{Theorem}

Regardless of the presence of noise in the circuit and for both the $\epsilon$-differential privacy and $(\epsilon,\delta)$-differential privacy, we believe that adjusting the number of shots is a more convenient way to achieve a specific privacy budget compared to other methods, like adding additional noise to the circuit or applying an extra error correction circuit.
\section{Numerical simulation}
In this section, we conduct simulations to analyze the relation between privacy budget $\epsilon$ and other variables in quantum circuits. We mainly discuss the impact of the number of shots and depolarizing noise parameters.

\subsection{The impact of the number of measurements}
Firstly, we perform the numerical simulations for $\epsilon$-differential privacy in the noiseless circuit. The default values of the variables used in our calculations are that $\mu=0.15, r=1,\beta=0.997$, where $\mu$ denotes the expected value of the circuit output, $r$ denotes the rank of projection observable, $\beta$ denotes the confident level. Figure. \ref{fig:eps} illustrates that as the number of shots increases, the level of $\epsilon$ also increases which means the privacy protection level is lower. With less number of shots, the noise of measurements is higher which can provide a higher privacy preservation level. It can be concluded that for smaller values of shots, ranging from 5 to 15, the quantum algorithm could offer relatively effective privacy protection.

\begin{figure}[htbp]
 \centering
    \includegraphics[width=0.42\textwidth]{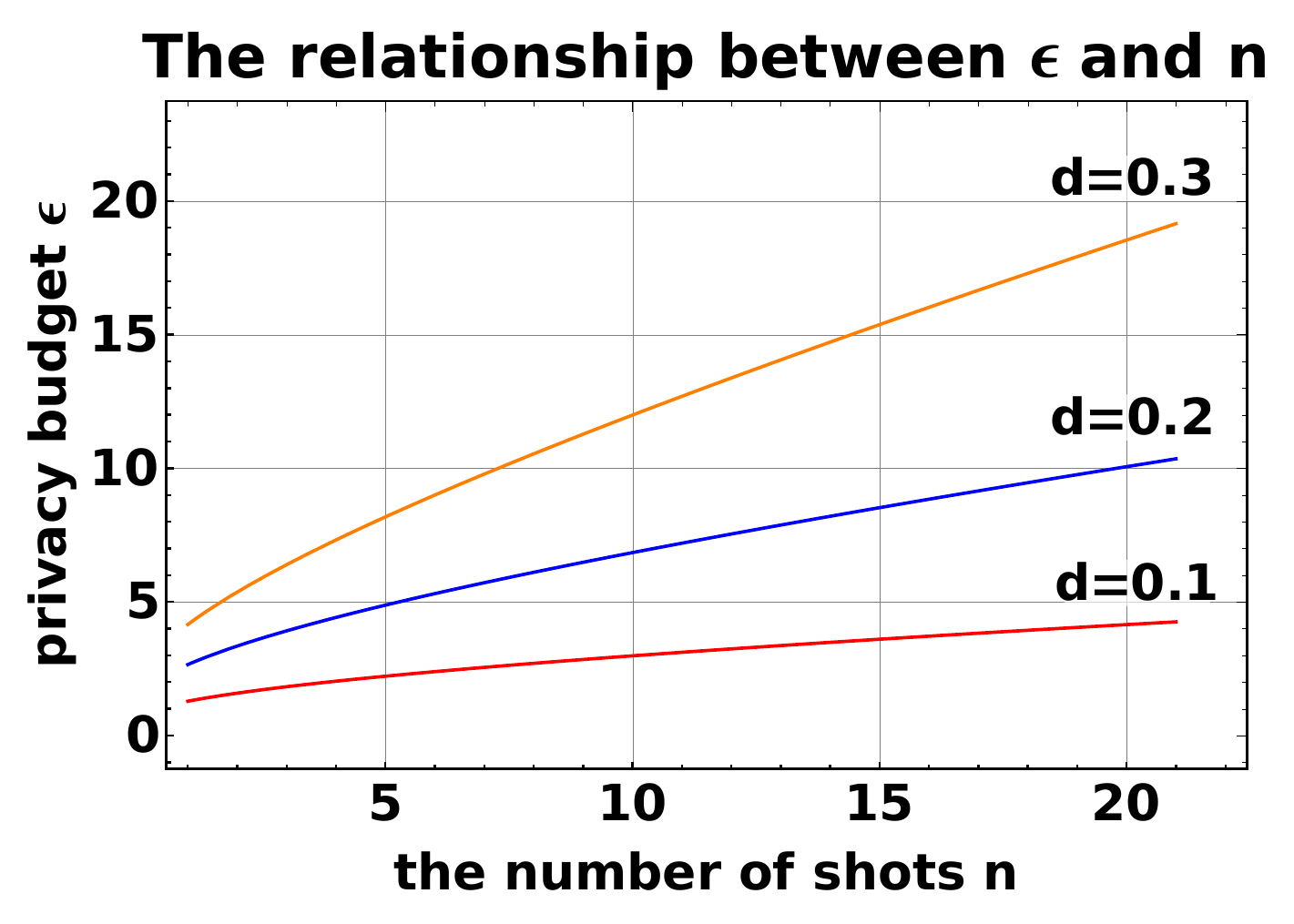}
  \caption{Fig. 3: $\epsilon-$differential privacy.}
  \label{fig:eps}
\end{figure}

\subsection{The impact of quantum noise}
Secondly, we will continue to perform numerical simulations for $\epsilon$-differential privacy with depolarizing noise. We start with depicting the relationship between privacy budget $\epsilon$ and noise parameter $p$. The default values used in our calculations are as follows: $\mu=0.15, r=1,\beta=0.997, D=2,n=10$. Additionally, we describe the relationship between $\epsilon$ and the number of shots $n$, using the default values: $\mu=0.15,r=1,\beta=0.997, D=2,d=0.1$, where $D$ denotes the dimension of Hilbert space, $r$ denotes the rank of observable. Fig. \ref{fig4} shows that as $p$ increases, $\epsilon$ decreases, which means that privacy protection is stronger. The relationship between $\epsilon$ and $n$ is consistent with the conclusion of the noiseless circuit.
\begin{figure}[htbp]
  \centering 
  \begin{subfigure}[b]{0.42\textwidth}
    \centering 
    \includegraphics[width=\textwidth]{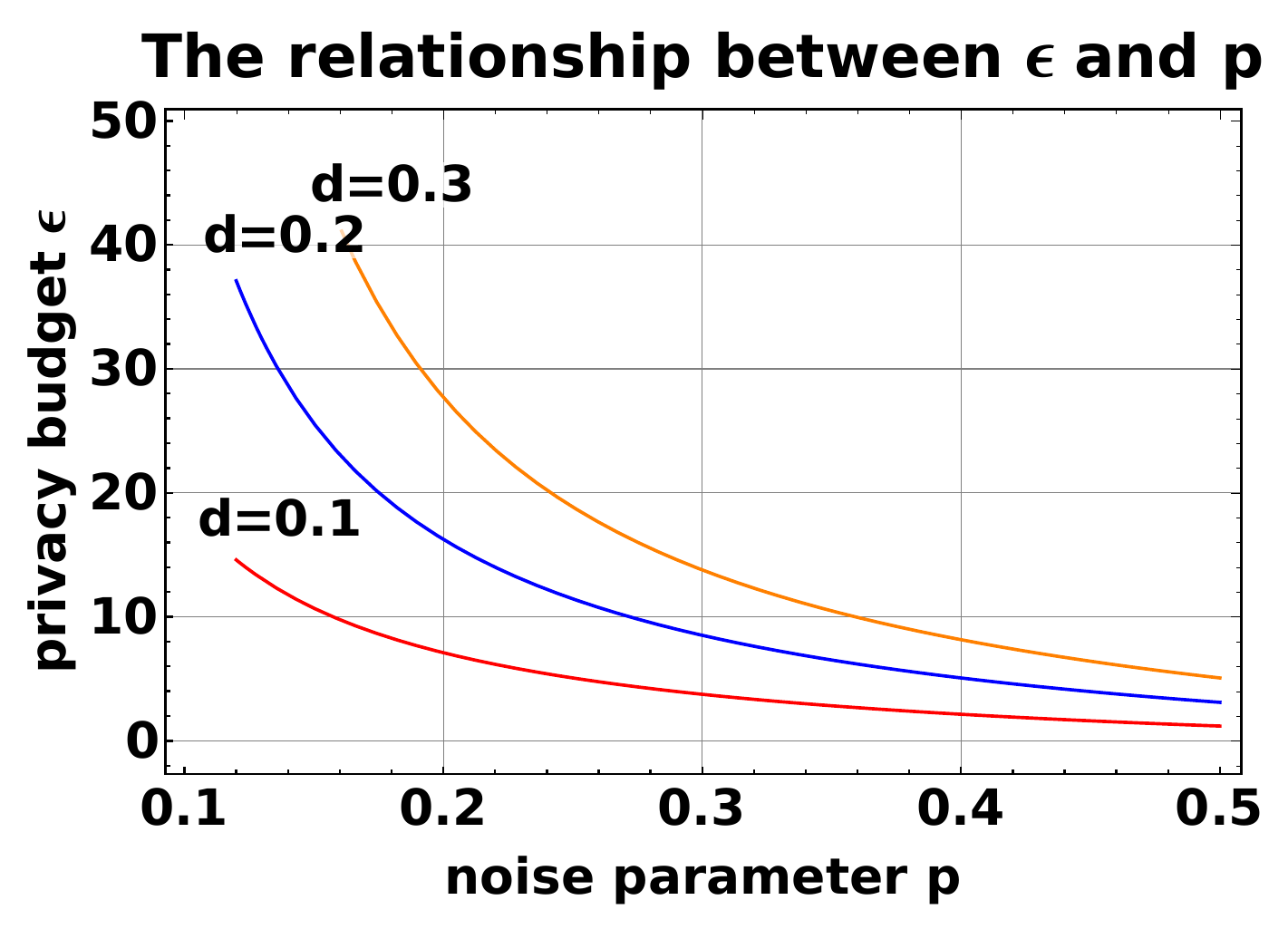}
  \end{subfigure}
  \hfill
  \begin{subfigure}[b]{0.42\textwidth}
    \centering 
    \includegraphics[width=\textwidth]{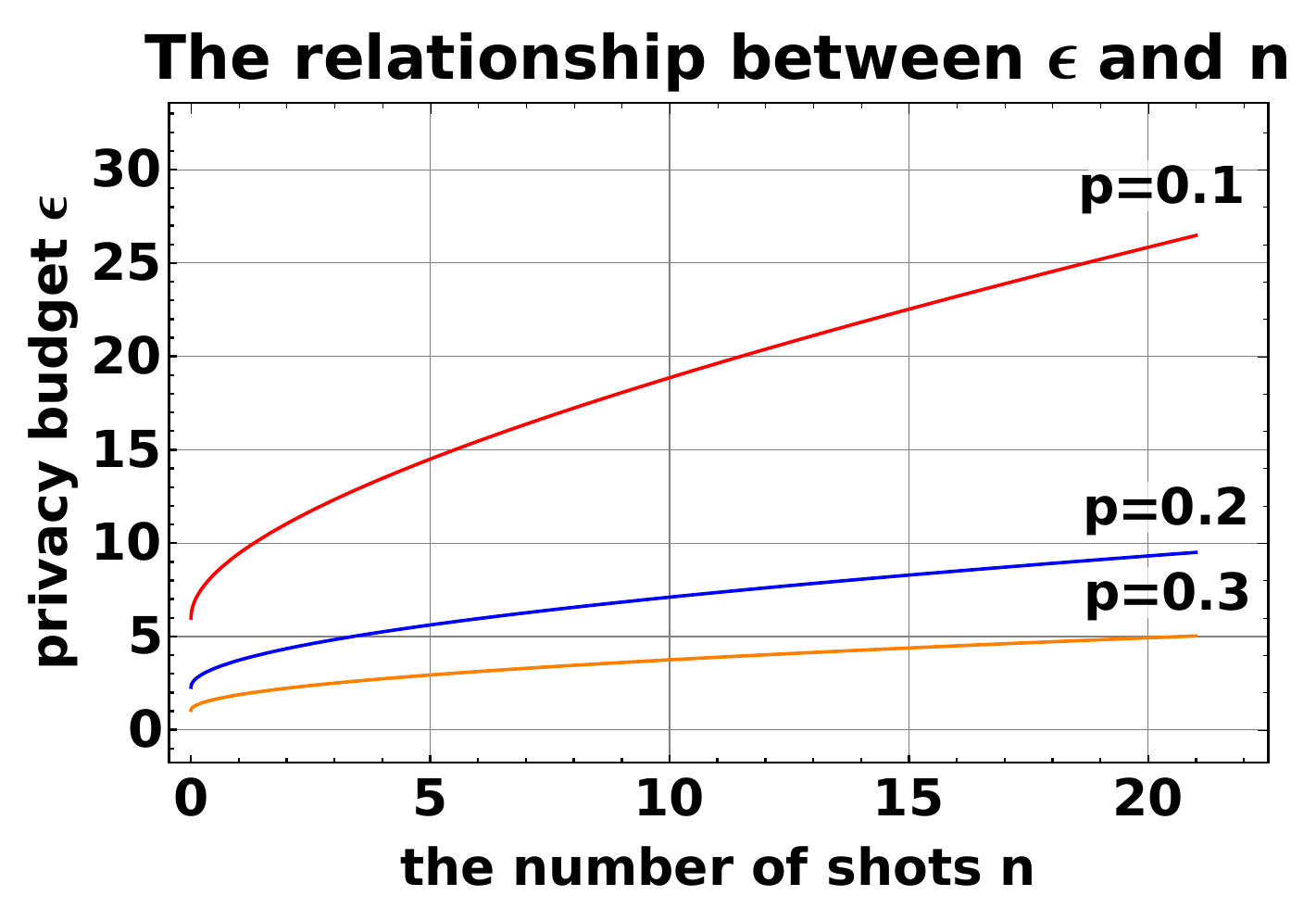}
  \end{subfigure}
  \caption{Fig. 4: $\epsilon-$differential privacy with depolarizing noise.}
  \label{fig4}
\end{figure}

\subsection{$(\epsilon,\delta)$-quantum differential privacy}
Here we discuss the factors affecting $(\epsilon, \delta)$-differential privacy and plot the relationship between $\epsilon$ and $\delta$ when $n$ is fixed, as well as the relationship between $\epsilon$ and $n$ when $\delta$ is fixed. The default parameters are selected to be consistent with the previous. 

Compared with Fig \ref{fig:eps}, the same value of $n$ results in significantly smaller privacy budget $\epsilon$. This shows that under a certain fault tolerance rate, privacy protection can be improved. Moreover, from Fig. \ref{fig:epsdelta}, when $n$ is constant, different values of $\delta$ will not cause too much difference in $\epsilon$.

\begin{figure}[htbp]
  \centering
  \begin{subfigure}[b]{0.42\textwidth}
    \centering
    \includegraphics[width=\textwidth]{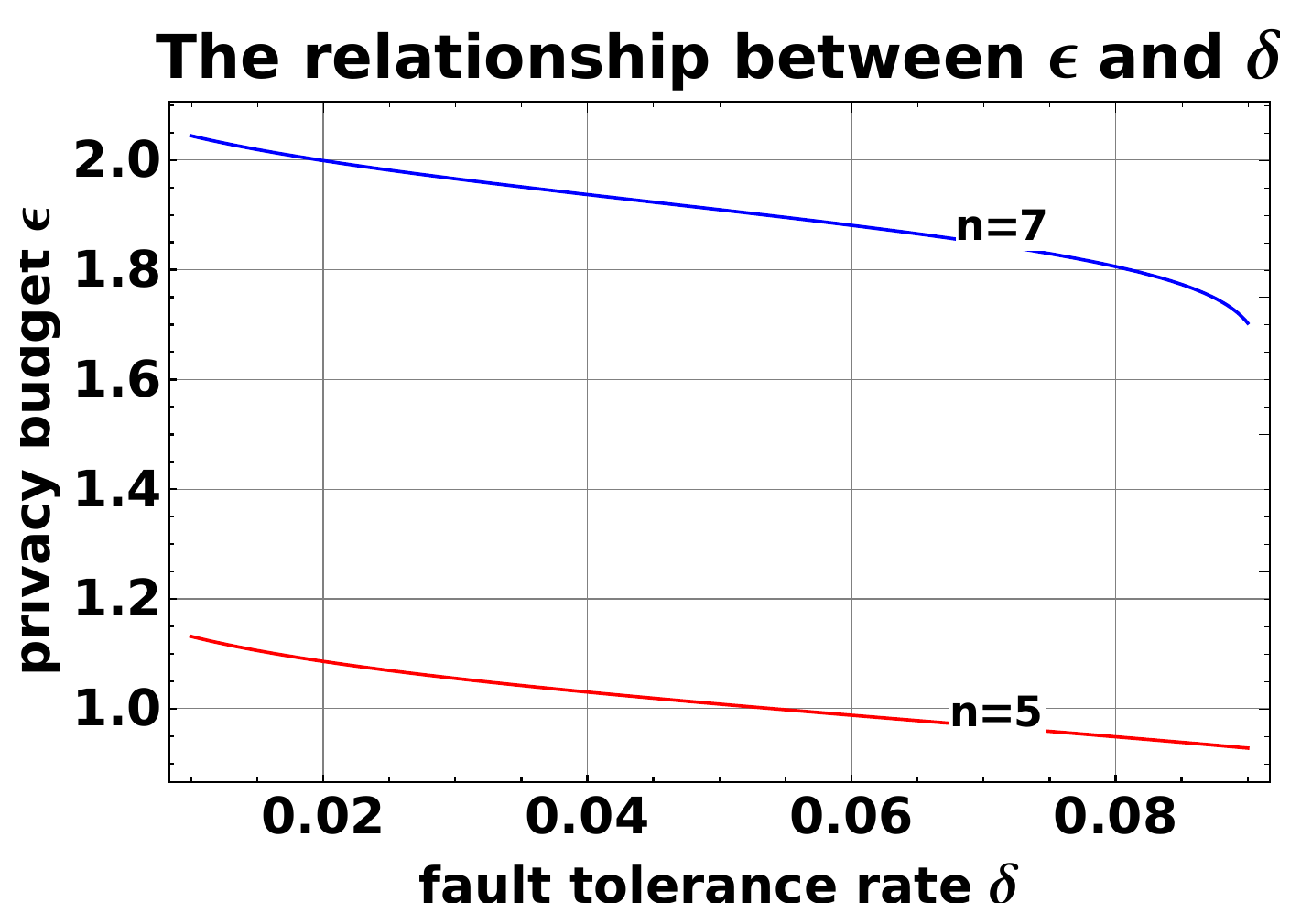}
  \end{subfigure}
  \begin{subfigure}[b]{0.42\textwidth}
    \centering
    \includegraphics[width=\textwidth]{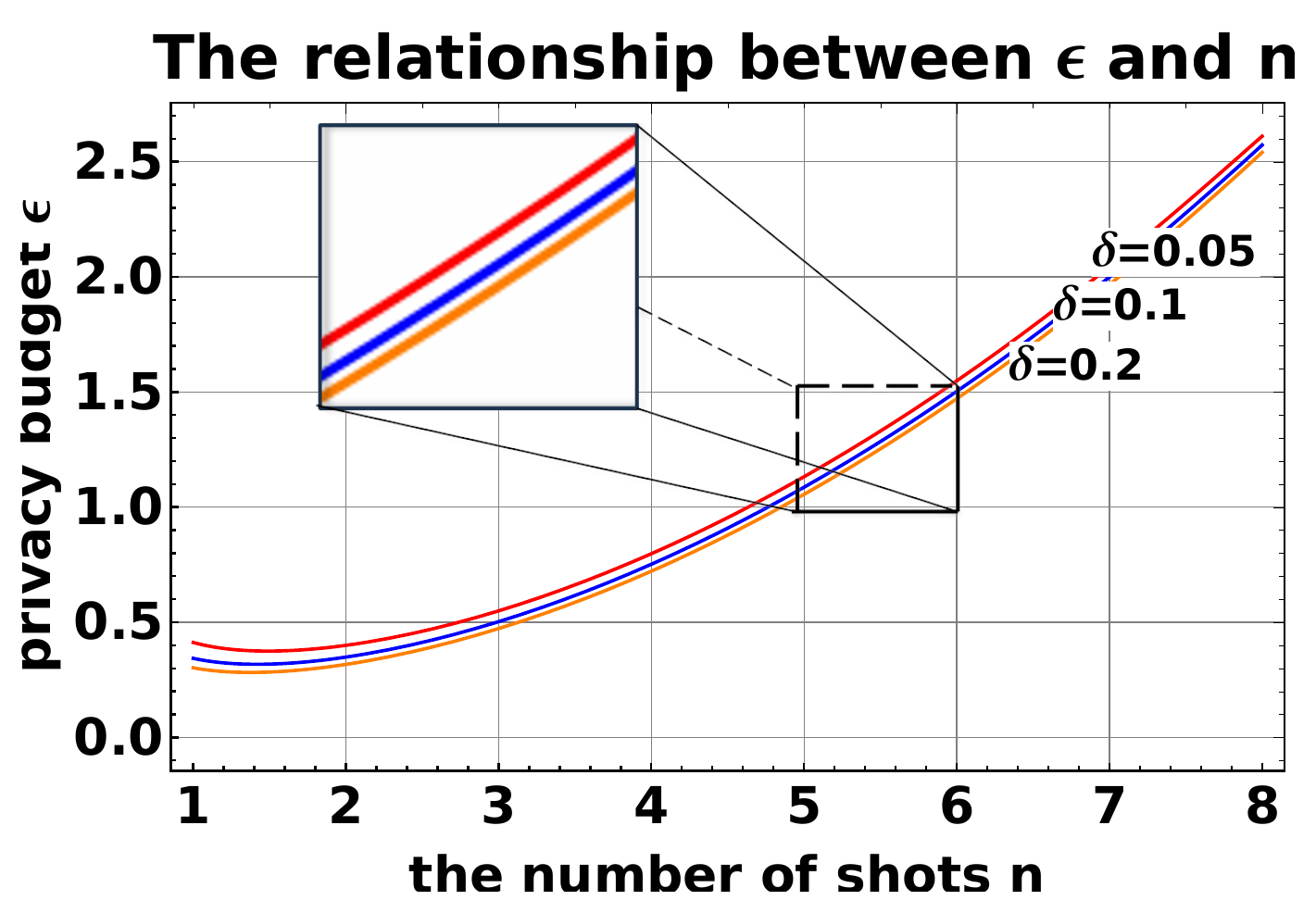}
  \end{subfigure}
  \caption{Fig. 5: $(\epsilon,\delta)-$differential privacy.}
  \label{fig:epsdelta}
\end{figure}
\section{CONCLUSION AND DISCUSSIONS}
In this paper, we obtain the relationship between shot noise and privacy budget. We specifically focus on two privacy notions: $\epsilon$-differential privacy and $(\epsilon,\delta)$-differential privacy. 
We believe that reducing the number of shots is more convenient than introducing additional noise into the circuit to achieve a specific privacy budget. Here are two reasons, on the one hand, the latter involves modifying the circuit, which may introduce more quantum gates during computation, thereby increasing the algorithm complexity. On the other hand, reducing the number of measurements directly reduces complexity and has nothing to do with the circuit structure.
In this paper, we mainly focus on theoretical calculations and numerical simulations, with no experiments conducted on quantum devices. Further experiments can be conducted to implement QDP in the future. Furthermore, it is worth exploring the potential to obtain a simple and neat mathematical expression of QDP.



\bibliographystyle{ieeetr}
\bibliography{cite}

\newpage

\end{document}